\documentclass{amsart}
\usepackage{amsmath,amssymb,amsthm,fullpage, mathrsfs, verbatim}

\def\ba #1\ea{\begin{align} #1 \end{align}}
\def\bas #1\eas{\begin{align*} #1 \end{align*}}
\newcommand{\R}{\mathbb{R}}

\newcommand{\C}{\mathbb{C}}

\newcommand{\h}{\mathcal{H}}
\newcommand{\A}{\mathcal{A}}

\newcommand{\Spec}{\textrm{Spec }}
\newcommand{\Hom}{\textrm{Hom}}
\newcommand{\g}{\mathfrak{g}}

\newtheorem{thm}{Theorem}[section]
\newtheorem{corr}[thm]{Corrollary}

\newtheorem{prop}[thm]{Proposition}

\theoremstyle{remark}
\newtheorem{remark}{Remark}
\theoremstyle{definition}
\newtheorem{dfn}{Definition}

\begin{document}

\title{Dynkin operators, renormalization and the geometric $\beta$ function}

\author{Susama Agarwala}
\address{Hamburg University}
	  \email{susama.agarwala@math.uni-hamburg.de}
 	  \urladdr{http://www.its.caltech.edu/$\sim$susama/} 

\begin{abstract}
In this paper, I show a close connection between renormalization and a
generalization of the Dynkin operator in terms of logarithmic
derivations. The geometric $\beta$ function, which describes the
dependence of a Quantum Field Theory on an energy scale defines is
defined by a complete vector field on a Lie group $G$ defined by a
QFT. It also defines a generalized Dynkin operator.

\end{abstract}

 \maketitle

\section{Introduction}
The Dynkin operator has recently become an important object in the
study of dynamical systems. The classical Dynkin operator defines a
bijection from a Lie group to it's Lie algebra, the inverse of the
exponential map. It is key in the closed form exapansion of the Baker-Campbell-Hausdorff formula. In \cite{MP}, the authors generalize Dynkin operators
in terms of logarithmic derivatives on a Lie algebra, and connect it
to Magnus-type formulas. The classical Magnus formula provides a
solution to the system of differential equations of the form \ba X'(t)
= A(t) X(t) \label{diffeq}\;.\ea Systems of this form appear in the
study of renormalization of quantum filed theories (QFTs). In
\cite{CK01}, the authors define a $\beta$ function, a Lie algebra
element representing how a dimensionally regularized QFT depends on
the energy scale. The $\beta$ function for dimensional regularization
and momentum cutoff regularization satisfies an equation of the form
\eqref{diffeq}. In \cite{CM06, thesis1}, the authors show that this
$\beta$ function defines a connection that also satisfies
\eqref{diffeq}. In this note, I show that there is a much deeper
connection between the Dynkin operator and renormalization.

As in the literature on the Hopf algebraic approach to
renormalization, intiated by \cite{CK00}, consider a regularized
perturbative Quantum Field Theory (QFT), $\phi$, as a map from Feynman
diagrams to an algebra $\A$. The divergence structure of the Feynman
diagrams is encoded in a Hopf algebra $\h$, as initially introduced by
Connes and Kreimer in \cite{CK00}.  I wish to keep the discussion in
the paper general, but for specific examples, one can consider the
Hopf algebra structure on scalar field theories, developed in
\cite{CK00}, on QED developed in \cite{vS07}, on gauge theories
developed in \cite{K06ag}. The algebras in all these cases have been
the algebra of formal Laurent series, $\A = \C[z^{-1}][[z]]$
\cite{cutoffdimreg}. However, if one is interested in momentum cut-off
renormalization, $\A = \C[\log z, z^{-1}][[z]]$ is appropriate
\cite{cutoffdimreg}. For a scalar field theory over a curved, compact
Euclidean background, use $\A =
\mathcal{D}'(M)[z^{-1},z]]$\cite{thesis2}.

In section \ref{geombeta}, I generalize the $\beta$ function defined
in \cite{CK01, thesis2, cutoffdimreg}. I generalize regularized
Feynman rules as elements of an affine Lie group associated to a Hopf
algebra, $\h$. The action of the renormalization scale generalizes to
a flow on this group. Specifically, it defines a one parameter family
of diffeomorphisms. The $\beta$ function defining the action of the
renormalization scale action is the vector field of the flow pulled
back to the Lie algebra. In section \ref{Dynkin}, I recall the Dynkin
operator, $D: T(V) \rightarrow V$, a map from a tensor algebra to the
underlying Lie algebra that defines a map from the Lie group $G =
exp(V)$ to $V$. In \cite{EGP}, the authors showed that the $\beta$
function of \cite{CK01} can be written as a variation of this map on
$G$. I generalize this map for the class of geometric $\beta$ function
defined in section \ref{geombeta}. In \cite{MP}, the authors define a
generalization of the classical Dynkin operator using logarithmic
derivatives with regards to a Lie derivative. I show that the
geometric $\beta$ function, as defined in \cite{cutoffdimred},
is compatible with the Dynkin variant defined in \cite{EGP}.

\section{The perturbative $\beta$ function\label{geombeta}}

The literature on renormalization theory is often confusing because of
different nomenclature referring to slightly different things in
different parts of the community. To avoid this confusion, I use this
section to set up a dictionary of what I mean when I use different terms
commonly found in the physics literature, and what mathematical
generalizations they correspond to. In this way, I motivate why the
definition of a geometric $\beta$ function is the appropriate object
of study.

\begin{dfn}
The Hopf algebra of Feynman diagrams, $\h$ is a commutative Hopf
algebra over a field $k$ of characteristic $0$ associated to the
Feynman diagrams for some QFT, as originally constructed in
\cite{CK00}. The Hopf algebra is constructed to encode the
subdivergence structure of the Feynman integrals in a manner that is
compatible with BPHZ renormalization.
\end{dfn}

Recall a few useful properties of a Hopf algebra of Feynman
diagrams. The Hopf algebra $\h$ is generated by all 1PI graphs of the
QFT. For more details on this Hopf algebra, see \cite{CK00, CMbook,
  thesis1}, The coproduct of a graph $\Gamma \in \h$ is given by \bas
\Delta(\Gamma) = \sum_{\substack{\gamma \subseteq \Gamma \\ \gamma,
    \Gamma/\gamma \in \h} } \gamma \otimes \Gamma/\gamma \;,\eas where
$\Gamma/\gamma$ is the obtained from $\Gamma$ by the contraction of
the connected components of $\gamma$ to a point. This coproduct
encodes the divergence structure found in BPHZ
renormalization. Multiplication of graphs is given by disjoint
union. The counit is written \bas \varepsilon ( h) = \begin{cases} h &
  \mbox{h } \in \h_0 \\ 0 & \mbox{else}. \end{cases}\eas The Hopf
algebra is graded by loop number, with the grading operator $Y
(\Gamma) = n \Gamma$ if $\Gamma$ has $n$ loops. The antipode is
defined recursively as \bas S(\Gamma) = -\Gamma -
\sum_{\substack{\gamma \subseteq \Gamma \\ \gamma, \Gamma/\gamma \in
    \h} } S(\gamma) \Gamma/\gamma \;.\eas

The coproduct structure on $\h$ induces a convolution product on the
associated affine group scheme $G = \Spec \h$. For a given $k$-algebra
$\A$, the Lie group $G(\A) = \Hom_{k\; alg}(\h, \A)$. That is, for $g,
g' \in G(\A)$ and $\Gamma \in \h$, \bas g \star g' (\Gamma) = (g
\otimes g') (\Delta \Gamma) \;.\eas Note that $\h \simeq k[G]$, the
ring of regular functions on $G$.

\begin{dfn}
In this paper, the renormalization group is $G$. It is the group of
evaluations of the Hopf algebra of the QFT $\h$.
\end{dfn}

Given a QFT, there are well established Feynman rules that assign a
divergent integral to each Feynman diagram. Given a regularization
scheme, the regularized Feynman rules assign to each diagram a
integral that evaluates into some algebra $\A$. This is a linear
map. If $\A$ is a $k$-algebra, the regularized Feynman rules define an
algebra homomorphisms from $\h$ to $\A$.

\begin{dfn}
The elements of $G(\A)$, with $G = \Spec {\h}$ are the generalized
regularized Feynman rules for a QFT.
\end{dfn}

Regularized Feynman rules can be written as elements of $G(\A)$ for
some appropriately defined $\A$. These are the physical regularization
theories. The general elements of the renormalization group $\phi \in
G(\A)$ need not have any physical interpretation at all.

\begin{dfn}
The renormalization mass scale of a physical theory is represented by
$\R_+$. It is the energy scale at which a physical theory is evaluated. In
this paper, I follow the convention of \cite{CM06} and complexify the
energy scale, and write it $e^s \in \C^\times$ for $s \in \C$.
\end{dfn}

The regularized Feynman integrals are functions of the renormalization
mass scale. 

\begin{dfn}
The renormalization scale action describes the dependence of the
generalized regularized theory on the renormalization mass scale.
\end{dfn}


For example, consider $\phi_{dr} \in G(\A)$, the dimensionally
regularized Feynman rules for an (integer) $d$-dimensional scalar
QFT. Let $z$ be a complex parameter. For a given diagram $\Gamma$ with
$I(\Gamma)$ internal edges and $L(\Gamma)$ loops, \bas
\phi_{dr}(z)(\Gamma) =A(d+z)^l \int_0^\infty\prod_{k=1}^{I(\Gamma)}
\frac{1}{f_k(p_i,e_j)^2 + m^2} \prod_{i=1}^{L(\Gamma)}p_i^{d+z-1} dp_i
\;.\eas Momentum cutoff regularization in the same theory gives \bas
\varphi_{mc}(z)(\Gamma) =
\int_{-\frac{1}{z}}^{\frac{1}{z}}\prod_{k=1}^{I(\Gamma)}
\frac{1}{f_k(p_i,e_j)^2 + m^2} \prod_{i=1}^{L(\Gamma)} d^dp_i \; .\eas
The action of the renormalization scale maps the momenta $p_i
\rightarrow e^s p_i$ and thus \bas \phi_{dr}(z) &\mapsto
e^{sYz}\phi_{dr}(z) \\ \phi_{mc}(z) &\mapsto \phi_{dr}(e^sz) \;.\eas
Dimensionally regularized Feynman rules are elements of the group
$\phi_{dr}(z) \in G(\C[z^{-1}][[z]])$. Momentum cutoff Feynman rules
are in $\phi_{mc} \in G(\C[z^{-1}, \log (z)][[z]])$.  For more details
on this renormalization scale action, see \cite{CMbook, cutoffdimreg}.

\begin{dfn}
The action of the renormalization scale on a physical regularized QFT,
$\phi \in G(\A)$ defines a one parameter path in $G(\A)$. This is
called the renormalization flow of $\phi$.\end{dfn}

The action of the renormalization scale on a particular physical
$\phi$ can be extended to an action of the renormalization scale on
$G(\A)$.

\begin{dfn}
Let $\sigma$ be an action of $\C$ on $G(\A)$ \bas \sigma : \C\times
G(\A) & \rightarrow G(\A) \\ (s, \phi) & \rightarrow \sigma(s)(\phi)
\;.\eas
\end{dfn}

In the examples above, I extend the dependence of dimensional
regularization and momentum cutoff regularization to an generalized
regularized theories as $\sigma_{dr}(s) (\phi) = e^{sYz}\phi(z)$ and
$\sigma_{mc}(s) (\phi) = \phi(e^sz)$.

For physical reasons, one expects the paths defined by the
renormalization scale to be integral; they are related to the
solutions of the renormalization group equations, which describe the
dependence of the observables of the theory on the energy scale. To
mimick this mathematically, I am interested in extensions of the
renormalization flows of physical theories to an action on $G(\A)$
such that for each $\phi \in G(\A)$, the renormalization flow,
$\sigma(s) \phi$ is an integral path in $G(\A)$. In other words the
renormalization group action on $G(\A)$ defines a one parameter family
of diffeomorphisms on $G(\A)$.

\begin{dfn}
An action $\sigma$ on $G(\A)$ defines a renormalization group flow if
it generates a one parameter family of diffeomorphisms on $G(\A)$.
\end{dfn}

For the next theorem, let $\A = \C[z^{-1}, \log(z)][[z]]$. Both
$\phi_{dr}$ and $\phi_{mc}$ can be written as elements of $G(\A)$.

\begin{prop}
The actions $\sigma_{dr}$ and $\sigma_{mc}$ both define one parameter
families of diffeomorphism on $G(\A)$.
\end{prop}

\begin{proof}
Let $* \in \{ mc, dr\}$. Since $\sigma_*$ is an action on
$G(\A)$, \begin{align*}\sigma_*(s)\circ \sigma_*(u)(\phi(z)) =
  \sigma_*(s+u) \phi(z) \;.\end{align*}

The action $\sigma_{dr}$ induces an automorphism on $G(\A)$
\cite{EM} \bas e^{sYz}(\phi \star \psi) = e^{sYz}\phi \star
e^{sYz}\psi \;. \eas Since the action is smooth, the result follows.

It is easy to check that $\sigma_{mc}$ is a smooth map. It remains to
check that it is bijective. To see surjectivity, notice that for any
fixed $s \in \C$ and any $\phi(z) \in G(\A)$, one can define
$\phi'_s(z) = \phi(e^{-s}z)$, and \bas \phi(z) = \sigma_{mc}(s)
\phi'_s(z)\;.\eas For injectivity, if there exists and $s \in \C$, and
$\phi, \psi \in G(\A)$, such that $\phi(sz) (\Gamma) =
\psi(sz)(\Gamma)$ for every $\Gamma \in \h$, then $\phi(z) (\Gamma) =
\psi(z)(\Gamma)$ for every $\Gamma \in \h$. This implies that $\phi(z)
= \psi(z)$.
\end{proof}

\begin{dfn}
The physical $\beta$ function for a renormalized QFT calculates the
dependence of the coupling constant on the renormalization scale \bas
\beta(g) = \frac{1}{\mu}\frac{d g}{d \mu} \;.\eas
\end{dfn}

The physical $\beta$ function is calculated perturbatively by loop number. In
this Hopf algebraic picture of renormalization, a related object
exists if the action $\sigma$ defines a renormalization group flow on
$G(\A)$. 

\begin{thm}
If $\sigma$ defines a renormalization group flow on $G(\A)$, it defines a
complete a vector field $X_\sigma \in \mathfrak{X}(G(\A))$.
\end{thm}

\begin{proof}
By hypothesis, $\sigma$ defines a one parameter family of
diffeomorphisms on $G(\A)$. Then $\sigma(s)\phi$ is an integral curve
in $G(\A)$ defined for all $s \in \C$, with $\sigma(0)\phi=
\phi$. Define a the vector field \bas X_\sigma(\sigma(s)\phi) =
\frac{d}{ds}\sigma(s)\phi\;.\eas This is is complete.
\end{proof}

\begin{dfn}
The geometric $\beta$ function for a renormalization group flow, $\sigma$,
is defined \bas \beta_\sigma: G(\A) & \rightarrow \g(\A) \\ \phi
& \rightarrow \phi^{-1} \star \frac{d}{ds}(\sigma(s)\phi) |_{s=0} =: \phi^{-1} \star X_\sigma(\phi)
\;.\eas
\end{dfn}

To see that $\beta_\sigma(\phi) \in \g(\A)$ for all $\phi \in G(\A)$,
note that $\beta_\sigma(\phi)$ is formed by left translating the
vector $X_{\sigma}(\phi) \in T_\phi G(\A)$ to $T_eG(\A) = \g$.

\begin{remark} 
In \cite{CK00}, the authors show that $z\beta_{\sigma_{dr}}(\phi) \in
\g(\C)$, and is the generator of the one parameter subgroup of $G(\A)$
defined $F_s(\phi) = \lim_{z\rightarrow 0}\phi^{-1}\star
\sigma_{dr}(s) \phi $. This is a happy accidental property of
dimensional regularization. It does not generalize to all
regularization schemes or regularization group actions.
\end{remark}

For more details on the geometric $\beta$ function, especially in the
case of $\sigma_{mc}$ and $\sigma_{dr}$, see \cite{cutoffdimreg}. In
the next section, we related the geometric $\beta$ function to the
Dynkin operator that appears in the study of dynamical systems.

\section{Generalized Dynkin operators and Geometric $\beta$ functions \label{Dynkin}}

Let $\mathcal{S}$ be a set and $k$ a field of characteristic $0$. Let
$V = k[\{\mathcal{S}\}]$ be the vector space generated by this
set. One can write $(V, [,])$ as a Lie algebra generated by $S$. The
$T(V)$, the tensor algebra on $V$, is the universal enveloping algebra
of $V$, $T(V) = \mathcal{U}(V)$. The classical (left) Dynkin operator
$D$ is a map \bas D : T(V) & \rightarrow (V, [,]) \\ x_1\otimes \cdots
\otimes x_n &\rightarrow [x_1, [\cdots,[x_{n-1},x_n]\cdots] \;.\eas
  Since $T(V) \simeq \mathcal{U}(V)$, $T(V)$ is isomorphic to a graded
  cocommutative Hopf algebra. Let $Y$ be the grading operator. The
  elements of $\mathcal{S}$ are primitive, which defines
  comultiplication. Multiplication is defined by concatenation. The
  antipode is defined \bas S(x_1\otimes \cdots \otimes x_n) =
  (-1)^nx_n\otimes \cdots \otimes x_1 \;.\eas Under this change of
  notation, the Dynkin operator $D = S \star Y$ \cite{Waldenfels} \bas
  S \star Y : \mathcal{U}(V) \rightarrow (V,[,]) \;.\eas The grading
  operator $Y$ is a derivation on $\mathcal{U}(V)$. Let $G =
  \textrm{exp}(V)$. The Baker-Campbell-Hausdorff (BCH) formula
  provides an inverse map from $G \rightarrow V$.  The Dynkin operator
  defines a closed form for the BCH formula \cite{EGP} \bas \log
  (\textrm{exp}X\textrm{exp}Y ) = \sum_{n>0}\frac {(-1)^{n-1}}{n}
  \sum_{ \begin{smallmatrix} {r_i + s_i > 0} \\ {1\le i \le
        n} \end{smallmatrix}} \frac{(\sum_{i=1}^n
    (r_i+s_i))^{-1}}{r_1!s_1!\cdots r_n!s_n!}  D( X^{r_1} Y^{s_1}
  X^{r_2} Y^{s_2} \ldots X^{r_n} Y^{s_n} )\;. \eas In fact, the Dynkin
  operator, $D$, defines a bijection from $G$ to $V$. I call this the
  Dynkin map.

In \cite{EGP}, the authors show that given any derivation $\delta$ on
a graded commutative Hopf algebra, $\h$ the map $D_\delta = S\star
\delta$ defines a bijection between $G(\A) = \Hom_{k\;alg}(\h, \A)$
and $\g(\A) = \textrm{Lie}(G(\A))$, by defining \bas D_\delta (\phi)
(\h) : = \phi (S \star \delta) ( \Delta (h)) \;. \eas It is easy to
check that the grading operator $Y$ is a derivation on $\h$. Using the
notation established in this paper, they show that
$z\beta_{\sigma_{dr}}$ corresponds to right composition by the Dynkin
map $D_Y = S \star Y$, \bas z\beta_{\sigma_{dr}}(\phi) = \phi^{-1} \star
Y\phi = \phi \circ D_Y \;. \eas  I generalize this
finding.

\begin{thm}
The geometric $\beta$ function, $\beta_\sigma$ is a generalized Dynkin
map, $D_{X_\sigma}$ from $G(\A)$ to $\g(\A)$ \bas
\beta_\sigma(\phi) = \phi^{-1} \star X_\sigma(\phi) = \phi
\circ D_{X_{\sigma}} \;. \eas \label{betaDynkin}
\end{thm}

\begin{proof}
The map $\beta_\sigma$ is defined by the vector field $X_\sigma$ on
$G(\A)$. Vector fields on a Lie group define derivations on the
algebra of regular functions on that group. Since $G = \Spec \h$, the
algebra of regular function $k[G] \simeq \h$. Therefore $X_\sigma$
defines a derivation on $\h$, call it $\delta_\sigma$. Specifically,
for $h \in k[G]$, \bas X_\sigma(\phi) \leftrightarrow
\delta_\sigma(h)(\phi) := \frac{d}{ds}h(\sigma_s(\phi))|_{s=0} \;.\eas
Recall that the product on $k[G]$ is defined pointwise \bas h h'(\phi)
= h(\phi) h'(\phi)\;. \eas It is easy to check that $\delta_\sigma$
is a derivation \bas \delta_\sigma(hh')(\phi) =
\frac{d}{ds}\left(hh'(\sigma_s(\phi))\right)|_{s=0} =
\frac{d}{ds}\left(h(\sigma_s(\phi))h'(\sigma_s(\phi))\right) =
\\ \frac{d}{ds}\left(h(\sigma_s(\phi))\right)|_{s=0}h'(\phi) +
h(\phi)\frac{d}{ds}\left(h(\sigma_s(\phi))\right)|_{s=0} =
(\delta_\sigma(h)h')(\phi) + (h\delta_\sigma(h'))(\phi) \;. \eas The
first equality is from the definition of $\delta_\sigma$, the second
from the definition of $k[G]$.  Under this set of definitions, \bas
\beta_\sigma(\phi)(h) = \phi^{-1}\star X_\sigma(\phi)(h) = \phi \circ
(S \star \delta_\sigma)(\Delta h)\;.\eas In other words, $\beta_\sigma
= \phi \circ D_{X_\sigma}$.
\end{proof}

\begin{remark}
Note that this implies that the geometric $\beta$ function
$\beta_\sigma$ defines a set bijection from $G(\A)$ to $\g(\A)$.
\end{remark}

\begin{corr}
The geometric $\beta$ function \bas \beta_\sigma : G(\A) \rightarrow
\g(\A)\;, \eas is defined by the Maurer-Cartan connection on the Lie
group $G(\A)$ contracted with $X_\sigma$.
\end{corr}

\begin{proof}
The Maurer-Cartan connection is a $\g(\A)$ valued one form defined
\bas \theta : \phi^{-1} \star d \phi \eas for $\phi \in
G(\A)$. Contracting with a vector field, $X_\sigma$ \bas \langle
X_\sigma(\phi) , \theta \rangle = \phi^{-1} \star X_\sigma(\phi) =
\beta_\sigma (\phi) \;. \eas
\end{proof}

In \cite{MP}, the authors define a generalization of the classical
Dynkin operator, $D_\delta = S \star \delta$ that is defined by a Lie
derivation $\delta$ on a free Lie algebra $\g$ \bas D_\delta :
\mathcal{U}(\g) \rightarrow \g \; .\eas It is a Lie idempotent in the
sense that if $x \in \g$, then $D_\delta (x) = \delta (x)$.

In the context of renormalization, the Hopf algebra of Feynman
diagrams $\h$, is of finite type. The Lie algebra $\g =
\textrm{Lie}(G)$ is freely generated, and the graded dual, $\h^\vee
\simeq \mathcal{U}(\g)$. Vector fields on $G$ exactly define Lie
derivatives on $\g$. This gives the following theorem.

\begin{thm}
The renormalization group flow defining action $\sigma$ defines a
generalized Dynkin operator in the sense of \cite{MP}.
\end{thm}

\begin{proof}
If the action $\sigma$ defines a renormalization group flow on
$G(\A)$, then it defines a one parameter family of diffeomorphisms on
$G(\A)$, and thus a complete vector field $X_\sigma \in
\mathfrak{X}(G(\A))$. The derivative $\delta_\sigma$ on $\h$ is exactly
the Lie derivative on $\g(\A)$ defined by $X_\sigma$.

The action $\sigma$ induces a path through $\g(\A)$ defined by the map
$\beta_\sigma$. Since $\beta_\sigma$ is a bijection from $G(\A)$ from
$\g(\A)$, for any $\alpha \in \g(\A)$, one can find a $\phi \in G(\A)$
such that $\alpha = \beta_\sigma(\phi)$. The action of $\sigma$ on
$G(\A)$ lifts to an action on $\g(\A)$ as \bas \sigma(s)(\alpha) =
\sigma(s)(\phi) \frac{d}{ds}(\sigma(s)(\phi)) \;.\eas The Lie
derivative $\delta_\sigma$ gives \bas \delta_\sigma(\alpha)(h) =
\alpha (\delta_\sigma(h)) = \frac{d}{ds}\sigma(s)(\alpha)(h) \;.  \eas
Let $\gamma \in \g(\A)$. Writing $\Delta(h) = \sum_{(h)} h_{(1)}
\otimes h_{(2)}$, \bas \delta_\sigma(\alpha\star \gamma)(h) =
\frac{d}{ds}(\sum_{(h)}\sigma(s)(\alpha)(h_{(1)})\sigma(s)(\gamma)(h_{(2)}))
\\ = \frac{d}{ds}(\sigma(s)(\alpha)) \star \sigma(s) (\gamma) (h) +
\sigma(s)(\alpha) \frac{d}{ds}(\sigma(s) (\gamma)) (h) =
\delta_\sigma(\alpha)\star \gamma (h) + \alpha \delta_\sigma(
\gamma)(h) \;.\eas
\end{proof}

To summarize, I relate the generalizations of the Dynkin map defined
in \cite{EGP} and the generalized Dynkin map defined in
\cite{MP}.

\begin{thm}
Each action $\sigma$ on $G(\A)$ that defines a one parameter family of
diffeomorphism on $G(\A)$ and thus induces the vector field
$X_\sigma$, defines a generalized Dynkin operator \bas D_{X_\sigma} :
\mathcal{U}(\g) \rightarrow \g \;.\eas The associated geometric
$\beta$ function $\beta_\sigma$ defines a generalized Dynkin map
defined by the Maurer-Cartan connection, $\theta$, \bas \beta_\sigma : G(\A) &
\rightarrow \g(\A) \\ \phi & \rightarrow <X_\sigma(\phi), \theta>
\;.\eas
\end{thm}

\section*{Acknowledgments}The author would like to thank Kurusch Ebrahimi-Fard,
  Joris Vankerschaver and Maria Barbero Li\~{n}an, for many useful
  discussions to develop the ideas in this paper, and is grateful for
  support from ICMAT via the Severo Ochoa Excellence grant to have
  these conversations.

\bibliographystyle{amsplain}
\bibliography{/home/mithu/bibliography/Bibliography}{}

\end{document}